\documentclass[11pt]{article} 
\usepackage[margin=1in]{geometry}
\usepackage[latin1]{inputenc}
\usepackage[english]{babel}
\usepackage[T1]{fontenc}
\usepackage{lmodern} 
\usepackage{graphicx}
\usepackage{wrapfig}
\usepackage{subfig}
\usepackage{caption}
\usepackage{paralist}
\usepackage{algorithm}
\usepackage{algorithmic}

\newcommand{\titel}{Edge-Orders}

\usepackage[usenames]{color}
\definecolor{hellblau}{rgb}{0.2,0.4,1} % eigene Farben definieren
\definecolor{dunkelblau}{rgb}{0,0,0.8}
\definecolor{dunkelgruen}{rgb}{0,0.5,0}
\usepackage[
	pdftex,
	colorlinks,
	linkcolor=dunkelblau,
	urlcolor=dunkelblau,
	citecolor=dunkelgruen,
	bookmarks=true,
	linktocpage=true,
	pdftitle={\titel},
	pdfauthor={},
	pdfsubject={}%,
]{hyperref}

\usepackage{amsmath}
\usepackage{amsthm} 
\usepackage{amsfonts} % Math-Befehle wie natürliche, reele Zahlen einbinden
\theoremstyle{plain} % kursiv
	\newtheorem{satz}{Satz}[] %section
	\newtheorem{theorem}[satz]{Theorem}
	\newtheorem{lemma}[satz]{Lemma}

\theoremstyle{remark} % normal
\theoremstyle{definition} % normal mit fettem Titel
	\newtheorem{definition}[satz]{Definition}
	
	\newtheorem*{conjecture}{Conjecture}

\usepackage{pdfpages} % um Appendix/PDFs einzubinden

\usepackage{etoolbox}
\newtoggle{long}
\toggletrue{long} % comment this to obtain the shortversion
\newcommand{\longversion}[1]{\iftoggle{long}{#1}{}}
\newcommand{\shortversion}[1]{\iftoggle{long}{}{#1}}

\begin{document}
	\title{\titel}
	%\author{}
	\author{Lena Schlipf\\LG Theoretische Informatik\\FernUniversität in Hagen, Germany
		\and Jens M. Schmidt\thanks{This research was supported by the DFG grant SCHM 3186/1-1.}\\Institute of Mathematics\\TU Ilmenau, Germany}
	%\keywords={}
	%\date{\renewcommand{\dateseparator}{.} \ddmmyyyydate \today}
	\date{}
	% footnote without mark
	%\makeatletter{\renewcommand*{\@makefnmark}{}\footnotetext{Email: \href{}{}}\makeatother}
	\maketitle

%\vspace{-1cm}
\begin{abstract}
Canonical orderings and their relatives such as st-numberings have been used as a key tool in algorithmic graph theory for the last decades. Recently, a unifying concept behind all these orders has been shown: they can be described by a graph decomposition into parts that have a prescribed vertex-connectivity.

Despite extensive interest in canonical orderings, no analogue of this unifying concept is known for edge-connectivity. In this paper, we establish such a concept named \emph{edge-orders} and show how to compute \emph{(1,1)-edge-orders} of 2-edge-connected graphs as well as \emph{(2,1)-edge-orders} of 3-edge-connected graphs in linear time, respectively. While the former can be seen as the edge-variants of \emph{st-numberings}, the latter are the edge-variants of \emph{Mondshein sequences} and \emph{non-separating ear decompositions}. The methods that we use for obtaining such edge-orders differ considerably in almost all details from the ones used for their vertex-counterparts, as different graph-theoretic constructions are used in the inductive proof and standard reductions from edge- to vertex-connectivity are bound to fail.

As a first application, we consider the famous \emph{Edge-Independent Spanning Tree Conjecture}, which asserts that every $k$-edge-connected graph contains $k$ rooted spanning trees that are pairwise edge-independent. We illustrate the impact of the above edge-orders by deducing algorithms that construct 2- and 3-edge independent spanning trees of 2- and 3-edge-connected graphs, the latter of which improves the best known running time from $O(n^2)$ to linear time.
\end{abstract}

\section{Introduction}
Canonical orderings serve as a fundamental tool in various fields of algorithmic graph theory, see~\cite{Badent2011,Schmidt2014} for a wealth of over 30 applications. Under this name, canonical orderings were published 1988 for maximal planar graphs~\cite{Fraysseix1988} and soon after generalized to 3-connected planar graphs~\cite{Kant1992}. Interestingly, it turned out only recently~\cite{Schmidt2014} that different research communities did, independently and partly even earlier, invent a strict generalization of canonical orderings to arbitrary 3-connected graphs under the names \emph{(2,1)-sequences}~\cite{Mondshein1971} (anticipating many of their later planar features in 1971) and \emph{non-separating ear decompositions}~\cite{Cheriyan1988}.

In particular,~\cite{Schmidt2014} exhibits a unifying view on all these different phrasings: In essence, a canonical ordering of a graph $G=(V,E)$ is a total order on $V$ such that for (almost) all $i$, the first $i$ vertices induce a 2-connected graph and the remaining vertices induce a connected graph in $G$. The ``heart'' of canonical orderings is thus \emph{connectivity}, with all of its implications for planarity, and not \emph{planarity} itself. For this reason, Mondshein called these orders (2,1)-sequences. In accordance with Biedl and Derka~\cite{Biedl2015a}, we therefore propose the general concept of $(k,l)$-\emph{orders}, which are total orders on $V$ whose prefixes induce $k$-connected and whose suffixes induce $l$-connected graphs.

Several relatives of canonical orderings aka (2,1)-orders fit into this context: The well-known $st$\emph{-numberings} and \emph{-orientations} are actually (1,1)-orders of 2-connected graphs, \emph{chain decompositions} are (2,2)-orders of 4-connected graphs, and further orders on restricted graph classes such as planar graphs and triangulations are known (see Table~\ref{tab:orders} left).

\begin{table}[h!t]
	\small
	\centering
	\subfloat{
		\begin{tabular}{|c||p{3.8cm}|p{5.1cm}|}
			\hline
			$k \backslash l$ & 1 & 2\\
			\hline
			\hline
			1 & $st$-numbering \cite{Even1976} $O(m)$ &\\
			\hline
			2 & Mondshein sequence \cite{Schmidt2014} $O(m)$ & Chain decomposition \cite{Curran2005} $O(n^2m)$, if planar: \cite{Nakano97} $O(m)$\\
			\hline
			3 & (3,1)-order \cite{Biedl2015a} for triangulations $O(m)$ & 5-canonical decomposition~\cite{Nagai2000} for triangulations $O(m)$\\
			\hline
			4 & &\\
			\hline
		\end{tabular}
		\label{tab:vertexorders}
	}
	\subfloat{
		\begin{tabular}{|c||p{3.4cm}|p{0.2cm}|}
			\hline
			$k \backslash l$ & 1 & 2\\
			\hline
			\hline
			1 & $st$-edge-numbering \cite{Annexstein1996} $O(m)$ \emph{(+in this paper)} &\\
			\hline
			2 & (2,1)-edge-order $O(m)$ &\\
			  & \emph{(in this paper)} &\\
			\hline
			3 & &\\
			\hline
			4 & &\\
			\hline
		\end{tabular}
		\label{tab:edgeorders}
	}
	\caption{Known $(k,l)$-orders for $(k+l)$-connected graphs along with the best-known running times for constructing them (\emph{left}), and $(k,l)$-edge-orders for $(k+l)$-edge-connected graphs (\emph{right}).}
	\label{tab:orders}
\end{table}

The purpose of this paper is to extent this unifying view further to $(k,l)$\emph{-edge-orders}, in which prefixes are $k$-edge-connected and suffixes are $l$-edge-connected. Despite the many known and heavily used vertex-orders above, their natural edge-variants do not seem to be well-studied. In fact, we are only aware of one technical report by Annexstein et al.~\cite{Annexstein1996}, which deals with (1,1)-edge-orders (under the name $st$-\emph{edge-orderings}). Besides this classification, we present a simple description how a (1,1)-edge-order can be computed. Our main contribution is then an algorithm that computes a (2,1)-edge-order of a 3-edge-connected graph in time $O(m)$ (see Table~\ref{tab:orders} right), of which the corresponding result for the vertex-counterpart took over 40 years.

From a top-level perspective, this latter result follows closely the proof outline used for its vertex-counterpart in~\cite{Schmidt2014}. However, each part of this proof requires new ideas and non-trivial formalizations: BG-sequences differ from Mader-sequences (and, although not too far apart, it took a 28-page paper to show that the latter can be computed in linear time as well~\cite{Mehlhorn2015}), both non-separateness and $\overline{G_i}$ differ considerably already in their definitions, and, here, we need last-values in addition to just birth-values.

Just like (2,1)-orders, which immediately led to improvements on the best-known running time for 5 applications~\cite{Schmidt2014,Biedl2015}, (2,1)-edge-orders seem to be an important and useful tool for many graph algorithms. We give one application, which is related to the edge-independent spanning tree conjecture~\cite{Itai1988} (further are in progress): By using a (2,1)-edge-order, we show how three edge-independent spanning trees of 3-edge-connected graphs can be computed in time $O(m)$, improving the best-known running time $O(n^2)$ by Gopalan et al.~\cite{Gopalan2011}.

After giving preliminary facts on ear decompositions, we explain the proposed linear-time algorithms for computing (1,1)- and (2,1)-edge-orders in Sections~\ref{sec:11edgeorder}--\ref{sec:computing}. Section~\ref{sec:spanningtrees} then shows algorithms for computing two and three edge-independent spanning trees.

\section{Preliminaries}
We use standard graph-theoretic terminology and consider only graphs that are finite and undirected, but may contain parallel edges and self-loops. In particular, cycles may have length one or two. For $k \geq 1$, let a graph $G$ be $k$-\emph{edge-connected} if $n := |V| \geq 2$ and $G$ has no edge-cut of size less than $k$.

\begin{definition}[\cite{Lovasz1985,Whitney1932a}]
An \emph{ear decomposition} of a graph $G=(V,E)$ is a sequence $(P_0,P_1,\ldots,P_k)$ of subgraphs of $G$ that partition $E$ such that (i) $P_0$ is a cycle that is no self-loop and (ii) every $P_i$, $1 \leq i \leq k$, is either a path that intersects $P_0 \cup \cdots \cup P_{i-1}$ in its endpoints or a cycle that intersects $P_0 \cup \cdots \cup P_{i-1}$ in a unique vertex $q_i$ (which we call \emph{endpoint} as well). Each $P_i$ is called an \emph{ear}. An ear is \emph{short} if it is an edge and \emph{long} otherwise.
\end{definition}

\begin{theorem}[\cite{Robbins1939}]\label{thm:robbins}
A graph is 2-edge-connected if and only if it has an ear decomposition.
\end{theorem}

According to Whitney~\cite{Whitney1932a}, every ear decomposition has exactly $m-n+1$ ears ($m := |E|$). For any $i$, let $G_i = (V_i,E_i) := P_0 \cup \cdots \cup P_i$ and $\overline{E_i} := E-E_i$. We denote the subgraph of $G$ that is induced by $\overline{E_i}$ as $\overline{G_i} = (\overline{V_i},\overline{E_i})$. Clearly, $\overline{G_j} \subset \overline{G_i}$ for every $i < j$. We note that this definition of $\overline{G_i}$ differs from the definition $\overline{G_i} := G-V_i$ that was used for (2,1)-vertex-orders~\cite{Schmidt2014}, due to the weaker edge-connectivity assumption.

For any ear $P_i$, let $inner(P_i) := V(P_i)-G_{i-1}$ be the set of \emph{inner vertices} of $P_i$ (for $P_0$, every vertex is an inner vertex). Hence, for a cycle $P_i \neq P_0$, $inner(P_i) = V(P_i)-q_i$. Every vertex of $G$ is an inner vertex of exactly one long ear, which implies that, in an ear decomposition, the inner vertex sets of the long ears partition $V$.

\begin{definition}
Let $D = (P_0,P_1,\ldots,P_{m-n})$ be an ear decomposition of $G$. For an edge $e$, let $birth_D(e)$ be the index $i$ such that $P_i$ contains $e$. 
For a vertex $v$, let $birth_D(v)$ be the index $i$ such that $P_i$ contains $v$ as inner vertex and let $last_D(v)$ be the maximal index $birth(vw)$ over all neighbors $w$ of $v$. Whenever $D$ is clear from the context, we will omit the subscript $D$.
\end{definition}

Thus, $P_{last(v)}$ is the last ear that contains $v$ and, seen from another perspective, the first ear $P_i$ such that $\overline{G_i}$ does not contain $v$. Clearly, a vertex $v$ is contained in $\overline{G_i}$ if and only if $last(v) > i$.

\section{The (1,1)-edge-order}\label{sec:11edgeorder}
Although (1,1)-edge-orders can be seen as edge-counterparts of $st$-numberings, they do not seem to be well-known. Let two edges be \emph{neighbors} if they share a common vertex. Annexstein et al.\ gave essentially the following definition.

\begin{definition}[{\cite{Annexstein1996}}]
Let $G=(V,E)$ be a graph with an edge $st$ that is not a self-loop. A \emph{(1,1)-edge-order} \emph{through} $st$ of $G$ is a total order $<$ on the edge set $E - st$ such that $m \geq 2$,
\begin{compactitem}
  \item[--] every edge $e$ that is not incident to $s$ has a neighbor $e'$ with $e' < e$ and
  \item[--] every edge $e$ that is not incident to $t$ has a neighbor $e'$ with $e < e'$.
\end{compactitem}
\end{definition}

Clearly, if $G$ has a (1,1)-edge-order through $st$, $G$ is 2-edge-connected, as neither $st$ nor any other edge can be a bridge of $G$ (note that this requires $m \geq 2$).

The converse statement was shown in~\cite{Annexstein1996} using a special type of ear decompositions based on breadth-first-search (however, without giving details of the linear-time algorithm). Here, we aim for a simple and direct (unlike, e.g., reducing to (1,1)-orders via line-graphs) exposition of the underlying idea and show that \emph{any} ear decomposition can be transformed to a (1,1)-edge-order in linear time.

For convenience, we use the \emph{incremental list order-maintenance problem}, which maintains a total order subject to the operations of (i) \emph{inserting} an element after a given element and (ii) \emph{comparing} two distinct given elements by returning the one that is smaller in the order. Bender et al.~\cite{Bender2002} show a simple solution for an even more general problem with amortized constant time per operation; we will call this the \emph{order data structure}.

\begin{lemma}\label{lem:11edgeorder}
Let $G$ be a 2-edge-connected graph with an edge $st$ that is not a self-loop. Then a (1,1)-edge-order through $st$ can be computed in time $O(m)$.
\end{lemma}
\begin{proof}
We compute an ear decomposition $D$ of $G$ such that $st \in P_0$. This can be done in linear time by any text-book-algorithm; see~\cite{Schmidt2013a} for a simple one. Let $<_0$ be the total order that orders the edges in $P_0-st$ consecutively from $s$ to $t$. Clearly, $<_0$ is a (1,1)-edge-order through $st$ of the 2-edge-connected graph $G_0$. We extend $<_{i-1}$ iteratively to a (1,1)-edge-order $<_i$ of $G_i$ by adding the next ear $P_i$ of $D$; then $<_{m-n}$ gives the claim.

The order itself is stored in the \emph{order data structure}. For every vertex $x$ in $G_{i-1}$, let $min(x)$ be the smaller of its two incident edges in $P_{birth(x)}$ with respect to $<_{i-1}$ (for later arguments, define $max(x)$ analogously as the larger such edge); clearly, $min(x)$ and $max(x)$ can be computed in constant time while adding $P_j$. When adding the ear $P_i$ with (not necessarily distinct endpoints) $x$ and $y$, let $e$ be the smallest edge in $\{min(x),min(y)\}$ with respect to $<_{i-1}$ (this needs amortized constant time by using at most one comparison of the data structure). Consider all edges of $P_i$ in consecutive order starting with a neighbor of $e$. We obtain $<_i$ from $<_{i-1}$ by inserting these edges as one consecutive block immediately after the edge $e$; this takes amortized time proportional to the length of $P_i$. Then the first edge of $P_i$ has a smaller neighbor in $<_i$ while the last has a larger neighbor in $<_i$ (for cycles $P_i \neq P_0$, this exploits that $q_i$ has another incident edge in $G_{i-1}$), which implies that $<_i$ is a (1,1)-edge-order.
\end{proof}

This (special) (1,1)-edge-order will allow for a very easy computation of two edge-independent spanning trees in Section~\ref{sec:spanningtrees} and serve as a building block for the computation of three such trees. If one wants to keep the root-paths in two edge-independent spanning trees short, a different (1,1)-edge-order~\cite{Annexstein1996} may be computed by maintaining $min(x)$ as the incident edge of $x$ that is minimal in $G_i$ in the above algorithm (this can be done efficiently by updating $min(x)$ whenever an ear with endpoint $x$ is added). However, the latter order cannot be used for three edge-independent spanning trees.

\section{The (2,1)-edge-order}
We define (2,1)-orders as special ear decompositions.

\begin{definition}\label{def:edge-order}
Let $G$ be a graph with distinct edges $rt$ and $ru$ ($t=u$ is possible). A \emph{(2,1)-edge-order} through $rt$ and avoiding $ru$ (see Figure~\ref{fig:ExampleOrder2}) is an ear decomposition $D$ of $G$ such that
\begin{compactenum}
	\item[1.] $rt \in P_0$,
	\item[2.] $P_{m-n} = ru$, and  $\hfill\triangleright$ i.e., the last ear is the short ear $ru$
	\item[3.] for every $0 \leq i < m-n$, $\overline{G_i}$ contains $inner(P_i)$ and, if $P_i$ is short, at least one endpoint of $P_i$.
\end{compactenum}
\end{definition}

\begin{figure}[h!t]
\centering
\includegraphics[scale=1.0]{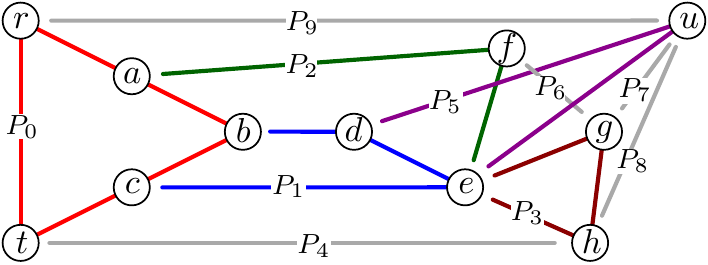}
\caption{A (2,1)-edge order of a $3$-edge connected graph.}\label{fig:ExampleOrder2}
\end{figure}

Definition~\ref{def:edge-order}.2 implies that $\overline{G_i}$ contains the vertices $r$ and $u$ for every $0 \leq i < m-n$. We call Definition~\ref{def:edge-order}.3 the \emph{non-separateness} of $D$. The non-separateness of $D$ states that every inner vertex of a long ear $P_i$ has an incident edge in $G$ that is in $\overline{G_i}$, and that every short ear $P_i$ (seen as edge) has a neighbor in $\overline{G_i}$. The name refers to the following helpful property.

\begin{lemma}\label{lem:non-separating}
Let $D$ be a (2,1)-edge-order. Then, for every $0 \leq i < m-n$, $\overline{G_i}$ is connected.
\end{lemma}
\begin{proof}
Consider any $i < m-n$ and let $e$ be any edge in $\overline{G_i}$. By Definition~\ref{def:edge-order}.2, $r \in \overline{G_i}$. We show that $\overline{G_i}$ contains a path from one of the endpoints of $e$ to $r$. This gives the claim, as $\overline{G_i}$ is an edge-induced graph and therefore does not contain isolated vertices.

Let $P_j$ be the unique ear that contains $e$. If $P_j$ is short, $P_j = e$ and $e$ has a neighbor in $\overline{G_j}$ due to the non-separateness of $D$. If $P_j$ is long, at least one endpoint of $e$ must be an inner vertex of $P_j$ and $e$ has a neighbor in $\overline{G_j}$ for the same reason. Hence, in both cases we find a neighbor that is contained in an ear $P_k$ with $k > j$. By applying induction on the indices of these ears, we find a path that starts with an endpoint of $e$ and ends with the only edge left in $\overline{G_{m-n-1}}$, namely $ru$.
\end{proof}

Next, we show that the existence of a (2,1)-edge-order proves the graph to be $3$-edge-connected.

\begin{lemma}\label{lem:3connected}
If $G$ has a (2,1)-edge-order, $G$ is 3-edge-connected.
\end{lemma}
\begin{proof}
Let $D$ be a (2,1)-edge-order through $rt$ and avoiding $ru$. Consider any vertex $v$ of $G$. By transitivity of edge-connectivity, it suffices to show that $G$ contains three edge-disjoint paths between $v$ and $r$. Let $P_i$ be the ear that contains $v$ as inner vertex. In particular $i < m-n$, as $P_i$ is long. Then $G_i$ has an ear decomposition and, due to Theorem~\ref{thm:robbins}, contains two edge-disjoint paths between $v$ and $r$. By Definitions~\ref{def:edge-order}.2+3, $\overline{G_i}$ contains $v$ and $r$. According to Lemma~\ref{lem:non-separating}, $\overline{G_i}$ is connected. Thus, $\overline{G_i}$ contains a third path between $v$ and $r$, which is edge-disjoint from the first two, as $G_i$ and $\overline{G_i}$ are edge-disjoint.
\end{proof}

Let $G$ have a (2,1)-edge-order. Then Lemma~\ref{lem:3connected} implies $\delta(G) \geq 3$. This in turn gives that, for every vertex $v$, $P_{last(v)}$ is not the first ear that contains $v$, which implies that $P_{last(v)}$ must have $v$ as endpoint. In particular, if $vw$ is an edge and $last(v)=last(w)=birth(vw)$, $P_{birth(vw)}$ is the short ear $vw$ and, according to the non-separateness of $D$, we have $i = m-n$, which implies $vw = ru$.

\begin{lemma}\label{lem:endpoint}
For any vertex $v$, $P_{last(v)}$ has $v$ as an endpoint. For any edge $vw$ satisfying $last(v)=last(w)=birth(vw)$, $vw = ru$.
\end{lemma}

The converse of Lemma~\ref{lem:3connected} is also true: If $G$ is 3-edge-connected, $G$ has a (2,1)-edge-order. This gives a full characterization of $3$-edge-connected graphs; however, proving the latter direction is more involved than Lemma~\ref{lem:3connected}. In the next section, we will prove the stronger statement that such a (2,1)-edge-order does not only exist but can actually be computed efficiently.

\section{Computing a (2,1)-edge-order}\label{sec:computing}
At the heart of our algorithm is the following classical construction of $3$-edge-connected graphs due to Mader.

\begin{definition}\label{def:maderoperation}
The following operations on graphs are called \emph{Mader-operations} (see Figure~\ref{fig:Mader}).
\begin{compactenum}[(a)]
	\item \emph{vertex-vertex-addition}: Add an edge between the not necessarily distinct vertices $v$ and $w$ (possibly a parallel edge or, if $v=w$, a self-loop).
	\item \emph{edge-vertex-addition}: Subdivide an edge $ab$ with a vertex $v$ and add the edge $vw$ for a vertex $w$.
	\item \emph{edge-edge-addition}: Subdivide two distinct edges $ab$ and $cd$ with vertices $v$ and $w$, respectively, and add the edge $vw$.
\end{compactenum}
\end{definition}

\begin{figure}[h!t]
	\centering
	\subfloat[vertex-vertex-addition: $v=w$ is allowed.]{
	\makebox[3.6cm]{
		\includegraphics[page=2,scale=0.6]{VertexVertex3}
		\label{fig:Mader1}
	}}
	\hspace{0.7cm}
	\subfloat[edge-vertex-addition: $w \in \{a,b\}$ is allowed.]{
		\includegraphics[page=1,scale=0.6]{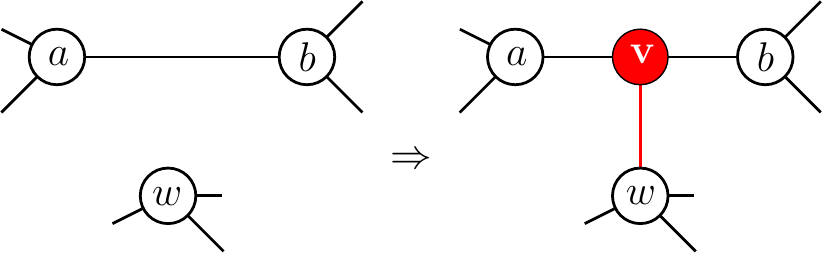}
		\label{fig:Mader2}
	}
	\hspace{0.7cm}
	\subfloat[edge-edge-addition: $a,b \in \{c,d\}$ is allowed.]{
		\includegraphics[page=1,scale=0.6]{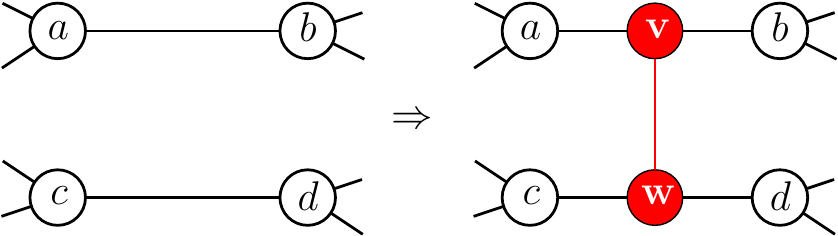}
		\label{fig:Mader3}
	}
	\caption{Mader-operations }
	\label{fig:Mader}
\end{figure}

The edge $vw$ is called the \emph{added edge} of the Mader-operation. Let $K_2^3$ be the graph that consists of exactly two vertices and three parallel edges.

\begin{theorem}[\cite{Mader1978}]\label{thm:mader}
A graph $G$ is $3$-edge-connected if and only if $G$ can be constructed from $K_2^3$ using Mader-operations.
\end{theorem}

According to Theorem~\ref{thm:mader}, applying Mader-operations on $3$-edge-connected graphs preserves $3$-edge-connectivity. We will call a sequence of Mader-operations that constructs a $3$-edge-connected graph a \emph{Mader-sequence}. It has been shown that a Mader-sequence can be computed efficiently. 

\begin{theorem}[{\cite[Thm.~4]{Mehlhorn2015}}]\label{thm:madersequence}
A Mader-sequence of a $3$-edge-connected graph can be computed in time $O(n+m)$.
\end{theorem}

Our algorithm for computing a (2,1)-edge-order works as follows. Assume we want a (2,1)-edge-order of $G$ through $r\overline{t}$ and avoiding $r\overline{u}$. We first compute a suitable Mader-sequence of $G$ using Theorem~\ref{thm:madersequence} and start with a (2,1)-edge-order of its first graph $K_2^3$. This can be easily found (see Figure~\ref{fig:K2_3}). The crucial part of the algorithm is then to iteratively modify the given (2,1)-edge-order to a (2,1)-edge-order of the next graph in the sequence efficiently.

\begin{figure}[h!t]
	\centering
	\includegraphics[scale=0.9]{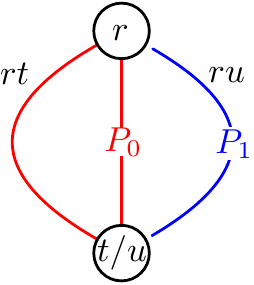}
	\caption{A (2,1)-edge-order of $K_2^3$ through $rt$ and avoiding $ru$.}
	\label{fig:K2_3}
\end{figure}

There are several technical difficulties to master. First, the edges $r\overline{t}$ and $r\overline{u}$ may change during the Mader-sequence, as they may be subdivided by Mader-operations. We therefore use a special Mader-sequence to harness the dynamics of the vertices $r$, $\overline{t}$ and $\overline{u}$. We choose a DFS-tree of $G$ with root $r$ and fix the edges $r\overline{t}$ and $r\overline{u}$. This way the initial $K_2^3$ will contain $r$ as one of its two vertices and, by the construction of~\cite[p.~6]{Mehlhorn2015}, $r$ will never be relabeled. Although $\overline{t}$ and $\overline{u}$ may not be present in this initial $K_2^3$, the bijection between the graphs $H$ of the Mader-sequence and $H$-subdivisions that are contained in $G$ as subgraphs~\cite[Thm.+Cor.~1]{Mehlhorn2015} (we refer to~\cite[Section~2.3]{Schmidt2010} for details of this bijection) gives us good replacement vertices $t$ and $u$ in $K_2^3$ for $\overline{t}$ and $\overline{u}$: If we, for every subdivision the Mader-sequence does on $rt$ or $ru$, respectively, label the subdividing vertex with $t$ or $u$ (the old $t$ or $u$ is then relabeled), the vertices labeled $t$ and $u$ at the end of the Mader-sequence will be $\overline{t}$ and $\overline{u}$. Thus, we can assume that the final (2,1)-edge-order is indeed through $r\overline{t}$ and avoids $r\overline{u}$, as desired. We refer to~\cite[Section~4]{Schmidt2010} for details on
how to efficiently compute such a labeling scheme.

Now consider a graph $G$ of the above Mader-sequence for which we know a (2,1)-edge-order $D$ and let $G'$ be the next graph in that sequence. Then $G'$ is only one Mader-operation away and we aim for an efficient modification of $D$ into a (2,1)-edge-order $D'$ of $G'$. We will prove that there is always a modification that is local in the sense that the only ears that are modified are ``near'' the added edge of the Mader-operation.

\begin{lemma}\label{lem:PathReplacement}
Let $D=(P_0,P_1,\ldots,P_{m-n})$ be a (2,1)-edge-order of a $3$-edge-connected graph $G$ through $rt$ and avoiding $ru$. Let $G'$ be obtained from $G$ by applying one Mader-operation $\Gamma$ and let $rt'$ and $ru'$ be the edges of $G'$ that correspond to $rt$ and $ru$ in $G$ (as discussed above). Then a (2,1)-edge-order $D'$ of $G'$ through $rt'$ avoiding $ru'$ can be computed from $D$ using only constantly many amortized constant-time modifications.
\end{lemma}

Lemma~\ref{lem:PathReplacement} is our main technical contribution and we split its proof into the following three sections. First, we introduce the operations $leg$, $belly$ and $head$ in order to combine several cases that can be handled similarly for the different types of $\Gamma$. Second, we show how to modify $D$ to $D'$ and, third, we discuss computational issues. 

For all three sections, let $vw$ be the added edge of $\Gamma$ such that $v$ subdivides the edge $ab \in E(G)$ and $w$ subdivides $cd \in E(G)$ (if applicable). Thus, the vertex $t'$ in $G'$ is either $t$, $v$ or $w$, and the vertex $u'$ in $G'$ is either $u$, $v$ or $w$ (hence, $t'r$ and $ru'$ will never be self-loops). In all three sections, $birth$ and $last$ will always refer to $D$, unless stated otherwise.

Let $P_i \neq P_0$ be an ear with a given orientation and let $x$ be a vertex in $P_i$. If $P_i$ is a path, we define $P_i[,x]$ and $P_i[x,]$ as the \emph{maximal subpaths} of $P_i$ that end and start at $x$, respectively; if $P_i$ is a cycle, we take the same definition with the additional restriction that $P_i[,x]$ starts at $q_i$ and $P_i[x,]$ ends at $q_i$. Occasionally, the orientation of $P_i$ will not matter; if none is given, an arbitrary orientation can be taken. For paths $A$ and $B$, let $A+B$ be the \emph{concatenation} of $A$ and $B$.

\subsection{Legs, bellies and heads}
While the operations $leg$ and $belly$ are inspired by the ones in~\cite{Schmidt2014}, the operation $head$ is new. All three operations will show for some special cases how $D$ can be modified to a (2,1)-edge-order $D'$. A complete description for all cases (using these operations) will be given in the next section.

\paragraph{Legs.} Let $\Gamma$ be either an edge-vertex-addition such that $last(w) < birth(ab)$ or an edge-edge-addition such that $birth(cd) < birth(ab)$.

  If $P_{birth(ab)}$ is long, at least one of $a$ and $b$ is an inner vertex, say w.l.o.g.\ $b$.
Otherwise, $P_{birth(ab)}=ab$ is short and, as $D$ is non-separating, at least one of $a$ and $b$, say w.l.o.g.\ $b$, has an incident edge in $\overline{G_{birth(ab)}}$. In both cases, orient $P_{birth(ab)}$ from $a$ to $b$. The operation $leg$ constructs $D'$ from $D$ by replacing the ear $P_{birth(ab)}$ of $D$ by the two consecutive ears $P_{birth(ab)}[,a] + av+ vw$ and $vb + P_{birth(ab)}[b,]$ in that order and, if $\Gamma$ is an edge-edge-addition, additionally subdividing the edge $cd$ in $P_{birth(cd)}$ with $w$ (see Figure~\ref{fig:leg}). Note that this definition is well-defined also for cycles $P_{birth(ab)}$, including self-loops.
 
\begin{figure}[h!t]
	\centering
	\includegraphics[scale=0.7]{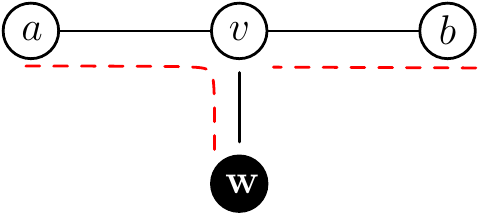}
	\caption{The result of operation $leg$ (dashed lines), black vertices are in $G_{birth(ab)-1}$.}
	\label{fig:leg}
\end{figure}

We prove that $D'$ is a (2,1)-edge-order through $rt'$ avoiding $ru'$. Assume first that $\Gamma$ is an edge-vertex-addition. Since $last(w) < birth(ab)$, we conclude that $w \notin P_{birth(ab)}$ ($w$ has no incident edge ``left'' in $\overline{G_{birth(ab)-1}}$). For the same reason, $birth(ab) > 0$. Hence, no matter whether $P_{birth(ab)}$ is a path or a cycle, $w$ and the one or two endpoints of $P_{birth(ab)}$ are contained in $G_{birth(ab)-1}$. Since $D'$ partitions $E(G')$, this implies that $D'$ is an ear decomposition. If $\Gamma$ is an edge-edge-addition, $birth(cd) \leq birth(ab)$ gives a very similar argument.

It remains to prove that $D'$ satisfies Properties~\ref{def:edge-order}.1--3. The first is true, as $rt \in P_0$ is only affected when $birth(cd) = 0$ and when $rt$ is subdivided by $w$; then $w = t'$ in $G'$ and $rt' \in P'_0$, as claimed. The second is true, as $cd \neq ru$ and, by assumption, $ab \neq ru$; hence, the last ear $ru$ does not change. For the non-separateness of $D'$, it suffices to consider the two modified ears $P_{birth(cd)}$ and $P_{birth(ab)}$, as all other ears still satisfy non-separateness. Since the only new inner vertex $w$ in $P'_{birth(cd)}$ is incident to the edge $wv \in \overline{G'_{birth(cd)}}$, $P'_{birth(cd)}$ is also non-separating. It remains to consider the two new ears $P'_{birth(ab)} = P_{birth(ab)}[,a]+av+ vw$ and $P'_{birth(ab)+1} = vb + P_{birth(ab)}[b,]$. All inner vertices of these ears except for the new vertex $v$ inherit their non-separateness directly from $P_{birth(ab)}$. Since $v$ is incident to $vb$, the long ear $P'_{birth(ab)}$ is non-separating and, if $P'_{birth(ab)+1}$ is long, $P'_{birth(ab)+1}$ is non-separating as well. If otherwise $P'_{birth(ab)+1}=vb$ is short, $P_{birth(ab)}$ cannot be long due to our assumed orientation. Hence, $P_{birth(ab)}=ab$ and the assumed orientation implies that $b$ has an incident edge in $\overline{G_{birth(ab)}}$, which gives that $P'_{birth(ab)+1}$ is non-separating as well.

\paragraph{Bellies.}
Let $\Gamma$ be either an edge-vertex-addition such that $last(w) = birth(ab)$ and $w\notin \{a,b\}$ or an edge-edge-addition such that $birth(cd) = birth(ab)$
(note that $c,d\in\{a,b\}$ is allowed.) %In both cases, $P_{birth(ab)}$ is a long ear; it contains at least $a$, $b$, $c$, or $d$ as inner vertex.
Consider the shortest path in $P_{birth(ab)}$ from an endpoint to one of the vertices $\{a,b\}$, say w.l.o.g.\ $b$, such that $w$ is contained in this path. We orient $P_{birth(ab)}$ from $a$ to $b$. $P_{birth(ab)}$ is a long ear with $b$ as inner vertex. If $\Gamma$ is an edge-edge-addition, one of the vertices $\{c,d\}$, say w.l.o.g.\ $c$, is contained in $P_{birth(ab)}[,w]$.

If $birth(ab)>0$, the operation $belly$ constructs $D'$ from $D$ by replacing the ear $P_{birth(ab)}$ of $D$ by the two consecutive ears $P_{birth(ab)}[,a]+av+ vw + P_{birth(ab)}[w,]$ and $vb+ P_{birth(ab)}[b,w]$ in that order (if edge-vertex-addition) and by the two consecutive ears $P_{birth(ab)}[,a]+av+ vw + wd+ P_{birth(ab)}[d,]$ and 
$vb+ P_{birth(ab)}[b,c]+cw$ (if edge-edge-addition), see Figure~\ref{fig:belly}. Note that this definition is well-defined also if $P_{birth(ab)}$ is a cycle.
If $birth(ab)=0$, the vertices $v$ and $w$ cut $P_0$ in two distinct paths $P_{0,1}$ and $P_{0,2}$ having endpoints $v$ and $w$.
Let $P_{0,1}$ be the path containing $r$. Then the operation $belly$ constructs $D'$ from $D$ by replacing the ear $P_{birth(ab)}$ of $D$ by the two consecutive ears $P_{0,1}+ vw$ 
and $P_{0,2}$ in this order. If $rt\in \{ab, cd\}$, then either $v=t'$ or $w=t'$, respectively.

\begin{figure}[h!t]
	\centering
	{
	\makebox[2cm]{
		\includegraphics[scale=0.7]{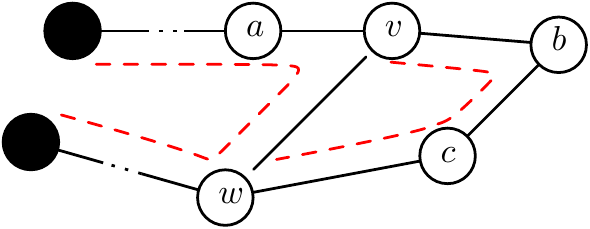}
	}}
	\hspace{4cm}
	{\makebox[2cm]{
		\includegraphics[scale=0.7]{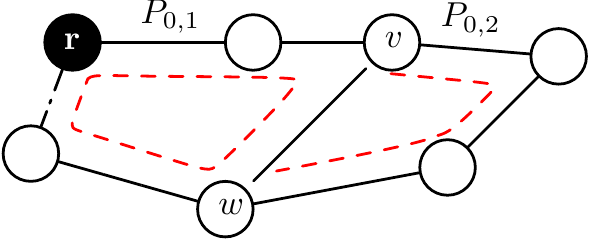}
	}}
	\caption{The result of the operation $belly$ (dashed lines). }
	\label{fig:belly}
\end{figure}

We prove that $D'$ is a (2,1)-edge-order through $rt'$ avoiding $ru'$.
No matter whether $P_{birth(ab)}$ is a path or a cycle, the one or two endpoints of it are contained in $G_{birth(ab)-1}$ and $D'$ partitions $E(G')$, so clearly $D'$ is an ear decomposition.

It remains to prove that $D'$ satisfies Properties~\ref{def:edge-order}.1--3. The first is true, as $rt \in P_0$ is only affected when $birth(ab) = 0$. Then,  if $rt$ is subdivided by $v$ or $w$,  $v=t'$ or $w = t'$ in $G'$, and $rt' \in P'_0$, as claimed. 
The second is true, as $ ru\notin \{ab,cd\}$ ($P_{birth(ab)}\neq \{ru\}$ as it is a long ear and $birth(cd)=birth(ab)$); hence, the last ear $ru$ does not change. 
For the non-separateness of $D'$, it again suffices to consider the modified ear $P_{birth(ab)}$. First, assume $birth(ab)>0$.
Consider the two new ears $P'_{birth(ab)} = P_{birth(ab)}[,a]+av+vw+ P_{birth(ab)}[w,]$ (respectively, $P'_{birth(ab)} = P_{birth(ab)}[,a]+av+vw+ wd +P_{birth(ab)}[d,]$ if edge-edge-addition) and $P'_{birth(ab)+1} = vb+ P_{birth(ab)}[b,w]$ (respectively, $P'_{birth(ab)+1} = vb+P_{birth(ab)}[b,c]+cw$ if edge-vertex-addition).
All inner vertices of these ears except for the new vertex $v$ (and $w$, if edge-edge-addition) inherit their non-separateness directly from $P_{birth(ab)}$.
Since $v$ is incident to $vb$ (and $w$ is incident to $wc$, if edge-edge-addition), the long ear $P'_{birth(ab)}$ is non-separating and $P'_{birth(ab)+1}$, which is long as it contains $b$ as inner vertex,  is non-separating as well.
If $birth(ab)=0$, very similar arguments show the non-separateness of the new ears. 

\paragraph{Heads.} 
Let $\Gamma$ be an edge-vertex-addition such that $w\in\{a,b\}$, say w.l.o.g.\ $w=a$, and $last(a) = birth(ab)$ and, if $ab=ru$, then $r\neq a$.
Then $a$ is an endpoint of $P_{birth(ab)}$ ($P_{birth(ab)}$ cannot be a self-loop, as $last(a)=birth(ab)$).
We orient $P_{birth(ab)}$ from $a$ to $b$. 
The operation $head$ constructs $D'$ from $D$ by replacing the ear $P_{birth(ab)}$ of $D$ by the two consecutive ears $av+va$ and $vb+P_{birth(ab)}[b,]$ in that order (see Figure~\ref{fig:head}). Note that this definition is well-defined also for cycles $P_{birth(ab)}$.

\begin{figure}[ht]
	\centering
	\makebox[2cm]{
		\includegraphics[scale=0.7]{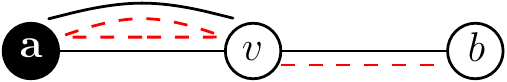}
	}
	\caption{The dashed lines show the result of the operation $head$.}
	\label{fig:head}
\end{figure}

We prove that $D'$ is a (2,1)-edge-order through $rt'$ avoiding $ru'$. 
Clearly, $D'$ is an ear decomposition.
 Property~\ref{def:edge-order}.1 is true, as $birth(ab)=last(a)>0$ and, hence, the first ear does not change. 
 Property~\ref{def:edge-order}.2 is true, as the last ear is only affected when $birth(ab) = ru$ and $r\neq a$; then $v=u'$ in $G'$ and 
the last ear in $D'$ is $ru'$, as claimed. 
For the non-separateness of $D'$, we consider the two new ears $P'_{birth(ab)} =av+va$ and $P'_{birth(ab)+1} =vb+P_{birth(ab)}[b,]$. $P'_{birth(ab)}$ is a long ear with $v$ as only inner vertex. Since $v$ is incident to $vb$, $P'_{birth(ab)}$ is non-separating.
All inner vertices of $P'_{birth(ab)+1}$ inherit their non-separateness directly from $P_{birth(ab)}$ and so, if $P'_{birth(ab)+1}$ is long, $P'_{birth(ab)+1}$ is non-separating as well. 
If otherwise $P'_{birth(ab)+1}=vb$ is short, then either $last(b)>last(a)$ and so $b$ has an incident edge in $\overline{G_{birth(ab)}}$, which gives that $P'_{birth(ab)+1}$ is non-separating as well.
If $last(b)=last(a)$ then $ab=ru$ (Lemma~\ref{lem:endpoint}) and the ear $P'_{birth(ab)+1}$ is the last ear of $D'$ and does not have to satisfy the non-separateness.

\subsection{Modifying D to D'}\label{sec:modify}
We will now show how to obtain a (2,1)-edge-order $D'$ through $rt'$ avoiding $ru'$ from $D$. By symmetry, assume w.l.o.g.\ that $birth(ab) \geq birth(cd)$. Note that applying the operations $belly$, $leg$ and $head$ preserves all properties of a $(2,1)$-edge-order. 
Recall that, for every subdivision the Mader-sequences does on $rt$ or $ru$, respectively, the subdividing vertex is $t'$ or $u'$, as explained after Figure~\ref{fig:K2_3}.
We have the following case distinctions:

\medskip
\begin{compactenum}
\item[\bf{1. $\mathbf{\Gamma}$ is a vertex-vertex-addition}] (see Figure~\ref{fig:Mader1})
\begin{compactenum}
	\item $vw$ is a self-loop at $v$ ($v=w$): Obtain $D'$ from $D$ by adding the new short ear $vv$ directly after the ear $P_{last(v)-1}$. This ensures that the new ear is non-separating. 
	\item $v\neq w$ and $vw\neq \{rt,ru\}$: If $last(v)\leq last(w)$,
$D'$ is obtained from $D$ by adding the new short ear $vw$ directly after the ear $P_{last(w)-1}$, ensuring that the new ear is non-separating. If $last(v)> last(w)$, the new short ear $vw$ is added directly after the ear $P_{last(v)-1}$.
		\item $vw=rt$ (the added edge is a parallel edge): the Mader-sequence gives us the information whether $rt$ is $rt'$ or the new added edge is $rt'$. 
	If $rt =rt'$ then add the new edge immediately after the ear $P_{last(t)-1}$. Otherwise obtain $D'$ from $D$ by replacing  $rt$ with $rt'$ in $P_0$ and adding the old edge $rt$ as an short ear immediately after the ear $P_{last(t)-1}$. 
		\item $vw=ru$ (the added edge is a parallel edge): 
the Mader-sequence gives us the information whether $ru$ is $ru'$ or the new added edge is $ru'$. 		
Depending on this information, obtain $D'$ from $D$ by either adding the  new edge directly before or directly after the last ear of $D$.  
\end{compactenum}
\item[\bf{2. $\mathbf{\Gamma}$ is an edge-vertex-addition}] (see Figure~\ref{fig:Mader2})
\begin{compactenum}
	\item $birth(ab)<last(w)$: Obtain $D'$ from $D$ by adding the new short ear $vw$ directly after the ear $P_{last(w)-1}$ and subdivide the ear $P_{birth(ab)}$ with $v$.
	This operation is also well-defined when $P_{birth(ab)}$ is a cycle or self-loop.
	Also, the new ear is non-separating and, since $v$ is incident to $w$, the ear $P_{birth(ab)}$ remains non-separating.
\item $last(w) < birth(ab)$ and $ab \neq ru$: Apply $leg$
	\item  $birth(ab)=last(w)$ and $w \notin \{a,b\}$: Apply $belly$.
\item $birth(ab)=last(w)$ and $w\in\{a,b\}$; if $ab=ru$, then $r\neq w$: Apply $head$. 
	\item $ab=ru$ and if $birth(ab)=last(w)$ and $w\in\{a,b\}$ then $r=w$:  Let $x\in\{a,b\}$ but $x\neq r$. Obtain $D'$ from $D$ by replacing the ear $ru$ by the two consecutive ears $wv+vx$ and $rv$. 
\end{compactenum}
\item[\bf{3. $\mathbf{\Gamma}$ is an edge-edge-addition}] (see Figure~\ref{fig:Mader3})
\begin{compactenum}
	\item $birth(ab)=birth(cd)$: Apply $belly$. 
	\item $birth(ab) > birth(cd)$ and $ab \neq ru$: Apply $leg$.
	\item $ab=ru$: Let w.l.o.g.\ $r=a$. Obtain $D'$ from $D$ by replacing the last ear of $D$ by the two consecutive ears $bv+vw$ and $rv$ in this order.
\end{compactenum}
\end{compactenum}

\medskip
In all cases, $D'$ is clearly an ear decomposition. 
Properties~\ref{def:edge-order}.1--3 are satisfied due to the given case distinction and the mentioned properties. Hence, $D'$ is a $(2,1)$-edge-order through $rt'$ avoiding $ru'$.

\shortversion{There are several subtleties in sorting out the computational complexity of this approach, mostly raised by the question how fast we can compute one of the above cases in which we are in. For a concise proof of the linear runtime, we refer to the appendix.}
\longversion{
\subsection{Computational complexity}
For proving Lemma~\ref{lem:PathReplacement}, it remains to show that each of the constantly many modifications above can be computed in constant amortized time. Note that ears may become arbitrarily long in the process and therefore may contain up to $\Theta(n)$ vertices. Moreover, we have to maintain the birth- and last-values in order to compute which subcase of the last section applies. Thus, we cannot use the standard approach of storing the ears of $D$ explicitly by using doubly-linked lists, as then the birth-values of linearly many vertices may change for every modification.

Instead, we will represent the ears as sets in a data structure for \emph{set splitting}, which maintains disjoint sets online under an intermixed sequence of find and split operations. Gabow and Tarjan~\cite{Gabow1985} discovered the first such data structure with linear space and constant amortized time per operation. Their and our model of computation is the standard unit-cost word-RAM. Imai and Asano~\cite{Imai1987} enhanced this data structure to an \emph{incremental variant}, which additionally supports adding single elements to certain sets in constant amortized time. In both results, all sets are restricted to be intervals of some total order. To represent the (2,1)-edge-order $D$ in the path replacement process, we will use the following more general data structure due to Djidjev~\cite[Section~3.2]{Djidjev2006}, which is not limited to total orders and still supports the add-operation.

\smallskip
The data structure maintains a collection $P$ of edge-disjoint paths under the following operations:

\begin{compactitem}
	\item[\texttt{new\_path(x,y)}:] Creates a new path that consists of the edge $xy$. The edge $xy$ must not be in any other path of $P$.
	\item[\texttt{find(e)}:] Returns the integer-label of the path containing the edge $e$.
	\item[\texttt{split(xy)}:] Splits the path containing the edge $xy$ into the two subpaths from $x$ to one endpoint and from $x$ to the other endpoint of that path.
	\item[\texttt{sub(x,e)}:] Modifies the path containing $e$ by subdividing $e$ with vertex $x$.
	\item[\texttt{replace(x,y,e)}:] Neither $x$ nor $y$ may be an endpoint of the path $Z$ containing $e$. Cuts $Z$ into the subpath from $x$ to $y$ and into the path that consists of the two remaining subpaths of $Z$ joined by the new edge $xy$.
	\item[\texttt{add(x,yz)}:] The vertex $y$ must be an endpoint of the path $Z$ containing the edge $yz$ and $x$ is either a new vertex or not in $Z$. Adds the new edge $xy$ to $Z$.
\end{compactitem}
\bigskip

Note that all ears are not only edge-disjoint but also internally disjoint. Djidjev proved that each of the above operations can be computed in constant amortized time~\cite[Theorem~1]{Djidjev2006}. We will only represent long ears in the data structure; the remaining short ears can be simply maintained as edges. As the data structure can only store paths, we store every cycle $P_i$ as the union of two paths in $P_i$ of which one is an edge with endpoint $q_i$ (for $P_0$, with endpoint $r$). For all paths of length at least two, including all long paths $P_i$, we store its two endpoints at its \texttt{find()}-label. Thus, the endpoints of all ears can be be accessed and updated in constant time.

This way, we store the ears of the initial (2,1)-edge-order of $K_2^3$ in constant total time. Every modification of Section~\ref{sec:modify} can then be realized with a constant number of operations of the data structure, and hence in amortized constant time.

Additionally, we need to maintain the order of the ears in $D$. Lemma~\ref{lem:PathReplacement} moves and inserts in every step only a constant number of ears to specified locations of $D$. Hence, we can maintain the order of ears in $D$ by applying the \emph{order data structure} (as defined for (1,1)-edge-oders) to the \texttt{find()}-labels of ears; this costs amortized constant time per step.

So far we could have maintained the order of ears also by using doubly-linked lists. However, for deciding which of the subcases in Section~\ref{sec:modify} applies, we additionally need to compare birth- and last-values of the vertices and edges involved in $\Gamma$. In fact, it suffices to support the queries ``$birth(x) < birth(y)$'' and ``$birth(x) = birth(y)$'', where $x$ and $y$ may be edges or vertices, and analogous queries on the $last$-values of vertices.
If $x$ and $y$ are edges, both birth-queries can be computed in constant amortized time by comparing the labels \texttt{find(x)} and \texttt{find(y)} in the order data structure. In order to allow birth-queries on vertices, we will store pointers at every vertex $x$ to the two edges $e_1$ and $e_2$ that are incident to $x$ in $P_{birth(x)}$. The desired query involving $birth(x)$ can then be computed by comparing \texttt{find(e$_1$)} in the order data structure.

For any new vertex $x$ that is added to $D$, we can find $e_1$ and $e_2$ in constant time, as these are in $\{av,vb,cw,wd,vw\}$. Since $P_{birth(x)}$ may change over time, we have to update $e_1$ and $e_2$. The only situation in which $P_{birth(x)}$ may loose $e_1$ or $e_2$ (but not both) is a \texttt{split} or \texttt{replace} operation on $P_{birth(x)}$ at $x$ (the split operation must be followed by an add operation on $x$, as $x$ is always inner vertex of some ear). This cuts $P_{birth(x)}$ into two paths, each of which contains exactly one edge in $\{e_1,e_2\}$. Checking \texttt{find(e$_1$)$=$find(e$_2$)} recognizes this case efficiently. Dependent on the particular case, we compute a new consistent pair $\{e'_1,e'_2\}$ that differs from $\{e_1,e_2\}$ in exactly one edge.
Finally, the value $last(x)$ for a vertex $x$ can be maintained the same way as $birth(x)$ with the only difference that it links to (one edge of) the last ear containing $x$ instead of the first such ear. This allows to check the desired comparisons in amortized constant time.

We conclude that $D'$ can be computed from $D$ in amortized constant time. This proves Lemma~\ref{lem:PathReplacement} and implies the following theorem.
}

\begin{theorem}\label{thm:main}
Given edges $tr$ and $ru$ of a $3$-edge-connected graph $G$, a (2,1)-edge-order $D$ of $G$ through $tr$ and avoiding $ru$ can be computed in time $O(m)$.
\end{theorem}

The proposed algorithms for (1,1)-edge-orders and (2,1)-edge-orders (as well as the computation of edge-independent spanning trees in the next section) are \emph{certifying} in the sense of~\cite{McConnell2011}: For (1,1)-edge-orders through $st$, it suffices to check that every edge $e \neq st$ has indeed a smaller and larger neighboring edge. For (2,1)-edge-orders, it suffices to check in linear time that $D$ is an ear decomposition of $G$ and that $D$ satisfies Definition~\ref{def:edge-order}.1--3.

\section{Edge-Independent Spanning Trees}\label{sec:spanningtrees}
Let $k$ spanning trees of a graph be \emph{edge-independent} if they all have the same root vertex $r$ and, for every vertex $x \neq r$, the paths from $x$ to $r$ in the $k$ spanning trees are edge disjoint. The following conjecture was stated 1988 by Itai and Rodeh.

\begin{conjecture}[Edge-Independent Spanning Tree Conjecture~\cite{Itai1988}]
Every $k$-edge-connected graph contains $k$ edge-independent spanning trees.
\end{conjecture}

The conjecture has been proven constructively for $k=2$~\cite{Itai1988} and $k=3$~\cite{Gopalan2011} with running times $O(m)$ and $O(n^2)$, respectively, for computing the corresponding edge-independent spanning trees. For every $k \geq 4$, the conjecture is open. We first give a short description of an algorithm for $k=2$ and then show the first linear-time algorithm for $k=3$.

For $k=2$, compute the (1,1)-edge-order $<$ through $tr$ using Lemma~\ref{lem:11edgeorder}. The first tree $T_1$ consists of the edges $min(x)$ for all vertices $x \neq r$ (as defined in Lemma~\ref{lem:11edgeorder}), while the second tree $T_2$ consists of $tr$ and the edges $max(x)$ for all vertices $x \notin \{r,t\}$. Then $T_1$ and $T_2$ are spanning, as no edge can be taken twice, and edge-independent, as, from every vertex $x$, the path of smaller edges to $r$ obtained by iteratively applying $min()$ must be edge-disjoint from the path of larger edges to $r$.

For $k=3$, choose any vertex $r$ and two distinct edges $tr$ and $ru$ in the 3-edge-connected graph $G$. Compute a (2,1)-edge-order $D$ through $tr$ and avoiding $ru$ in time $O(m)$ using Theorem~\ref{thm:main}. For every vertex $x \in V$, the idea is now to find \emph{two} edge-disjoint paths from $x$ to $r$ in $G_{birth(x)}$ (after all, $G_{birth(x)}$ is 2-edge-connected and thus contains a (1,1)-edge-order) and a \emph{third} path from $x$ to $r$ in $\overline{G_{birth(x)}}$ using the non-separateness of $D$. The subtle part is to make this idea precise: We have to construct the first tree $T_1$ in such a consistent way that the paths of smaller edges from $x$ to $r$ for \emph{all} vertices $x \in V$ are contained in $T_1$ (and the same for $T_2$ and paths of larger edges).

For a (1,1)-edge-order $<$ through $tr$ of $G$, let a spanning tree $T_1 \subseteq G$ be \emph{down-consistent} to a given (2,1)-edge-order through $tr$ if (a) every path in $T_1$ to $r$ is strictly decreasing in $<$ and (b) for every $0 \leq i \leq m-n$, $T_1 \cap G_i$ is a spanning tree of $G_i$ (analogously, \emph{up-consistent} spanning trees $T_2$ of $G-r$ are defined by strictly increasing paths to $t$). Now let a (1,1)-edge-order be \emph{consistent} to a given (2,1)-edge-order $D'$ if $G$ contains $r$-rooted spanning trees $T_1$ and $T_2$ that are down- and up-consistent to $D'$, respectively. By the very same argument as used for $k=2$, $T_1$ and $T_2 + tr$ are edge-independent and, in addition, do not use any edge of $\overline{G_{birth(x)}}$ for any $x \in V$.

In fact, the special (1,1)-edge-order that is computed by Lemma~\ref{lem:11edgeorder} is consistent to $D$: There, the trees $T_1$ and $T_2$ consist of the edges $min(x)$ and $max(x)$ for $x \in V$, which makes $T_1$ down-consistent and $T_2+tr$ up-consistent to $D$ (see Figure~\ref{fig:Consistent}). We note that the simpler and more established definition of consistent (1,1)-edge-orders~\cite{Cheriyan1988} as \emph{orders that remain (1,1)-edge-orders for all subgraphs} $G_i$, $0 \leq i \leq m-n$, does not suffice here (see Figure~\ref{fig:NotConsistent}).

\begin{figure}[h!t]
	\centering
	\subfloat[A consistent order $<$ and the resulting three edge-independent spanning trees.]{
		\includegraphics[page=1,scale=1.0]{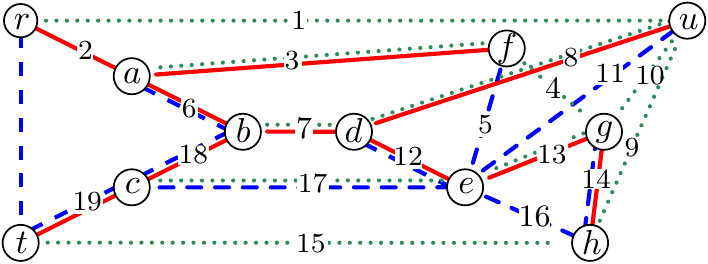}
		\label{fig:Consistent}
	}
	\hspace{0.5cm}
	\subfloat[Although $<$ is a (1,1)-edge-order for every $G_i$, $0 \leq i \leq m-n$, $<$ is not consistent: Any down-consistent tree contains the root-paths $12,11,10,2$ in $G_2$ and $6,5,3,2$ in $G_5$, which implies a cycle.]{
		\includegraphics[page=1,scale=1.0]{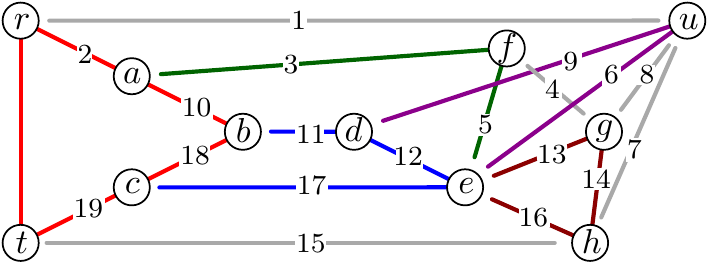}
		\label{fig:NotConsistent}
	}
	\caption{(1,1)-edge-orders that are consistent and not consistent to the (2,1)-edge-order of Figure~\ref{fig:ExampleOrder2}.}
	\label{fig:Consistence}
\end{figure}

It remains to construct the third edge-independent spanning tree. For every edge $e \neq ru$ of $G$, we compute a pointer to an arbitrary neighboring edge $e'$ in $\overline{G_{birth(e)}}$. This edge $e'$ exists, as $D$ is non-separating, and satisfies $birth(e') > birth(e)$. Similarly, for every vertex $x \in V-r-u$, we compute a pointer to an incident edge $e'$ of $x$ with $birth(e') > birth(x)$. Both computations take linear total time by comparing $birth$ values. The third edge-independent spanning tree is then the union of $ur$ and the $u$-rooted spanning tree of $G-r$ that interprets the pointers as parent edges. Hence, three edge-independent spanning trees can be computed in time $O(m)$.

\shortversion{One could be interested in the reason why, e.g., the reduction from $k$-edge- to $k$-vertex-connectivity by Galil and Italiano~\cite{Galil1991} cannot be applied in order to obtain edge-independent spanning trees from their vertex-counterparts; we give a reason in the appendix.}
\longversion{
\paragraph{Relation to vertex-independent spanning trees.}
The conjecture above has also received considerable attention for the vertex-case. Recently, a linear-time algorithm for computing three vertex-independent spanning trees of a 3-connected graph was given by~\cite{Schmidt2014}.
One could be interested in the reason why, e.g., the reduction from $k$-edge- to $k$-vertex-connectivity by Galil and Italiano~\cite{Galil1991} cannot be applied to modify the $3$-edge-connected input graph $G$ to a $3$-connected one such that three vertex-independent spanning trees for the latter give three edge-independent spanning trees in $G$. 
The reason is that, although such a reduction attempt is able to give three edge-disjoint paths between two given vertices, for multiple vertex pairs, the union of these paths may form cycles (see Figure~\ref{fig:exGI}). 

\begin{figure}[h!tb]
	\centering
	\subfloat[A 3-edge-connected graph $G$.]{
	\makebox[5cm]{
		\includegraphics[scale=0.6]{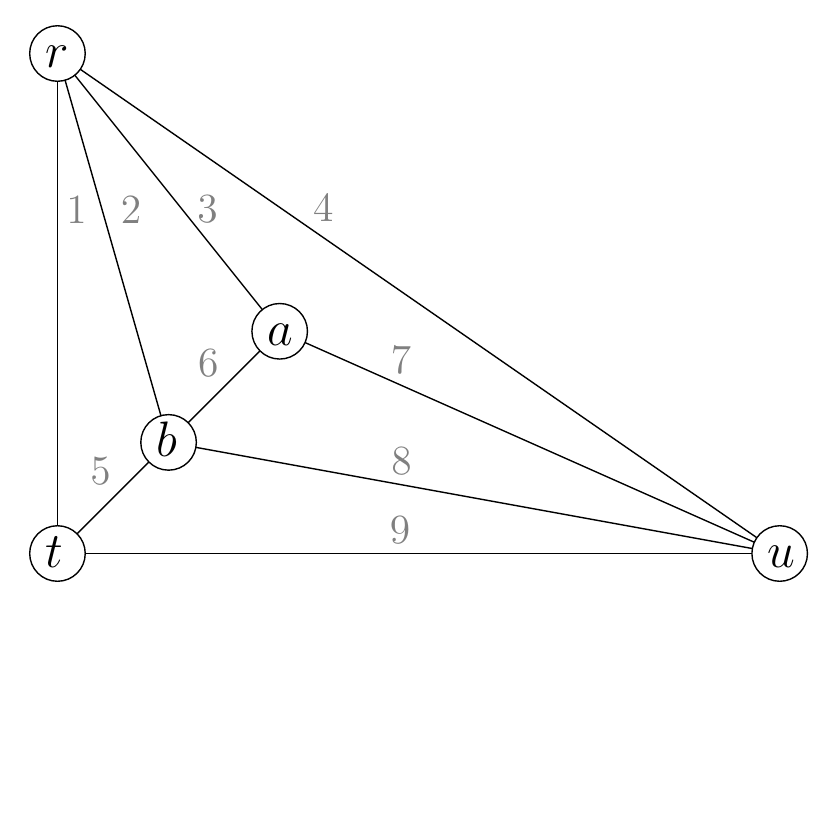}
		\label{fig:exGI_1}
	}}
	\hspace{1cm}
	\subfloat[The 3-connected graph $\hat{G}$ to which $G$ is reduced to using~\cite{Galil1991}, and a (2,1)-order of $\hat{G}$ through $r\hat{t}$ avoiding $\hat{u}$ ($t$ and $u$ have to be replaced, as $\hat{G}$ does not contain $rt$ and $ru$ anymore). Gray lines depict short ears.]{
		\includegraphics[scale=0.4]{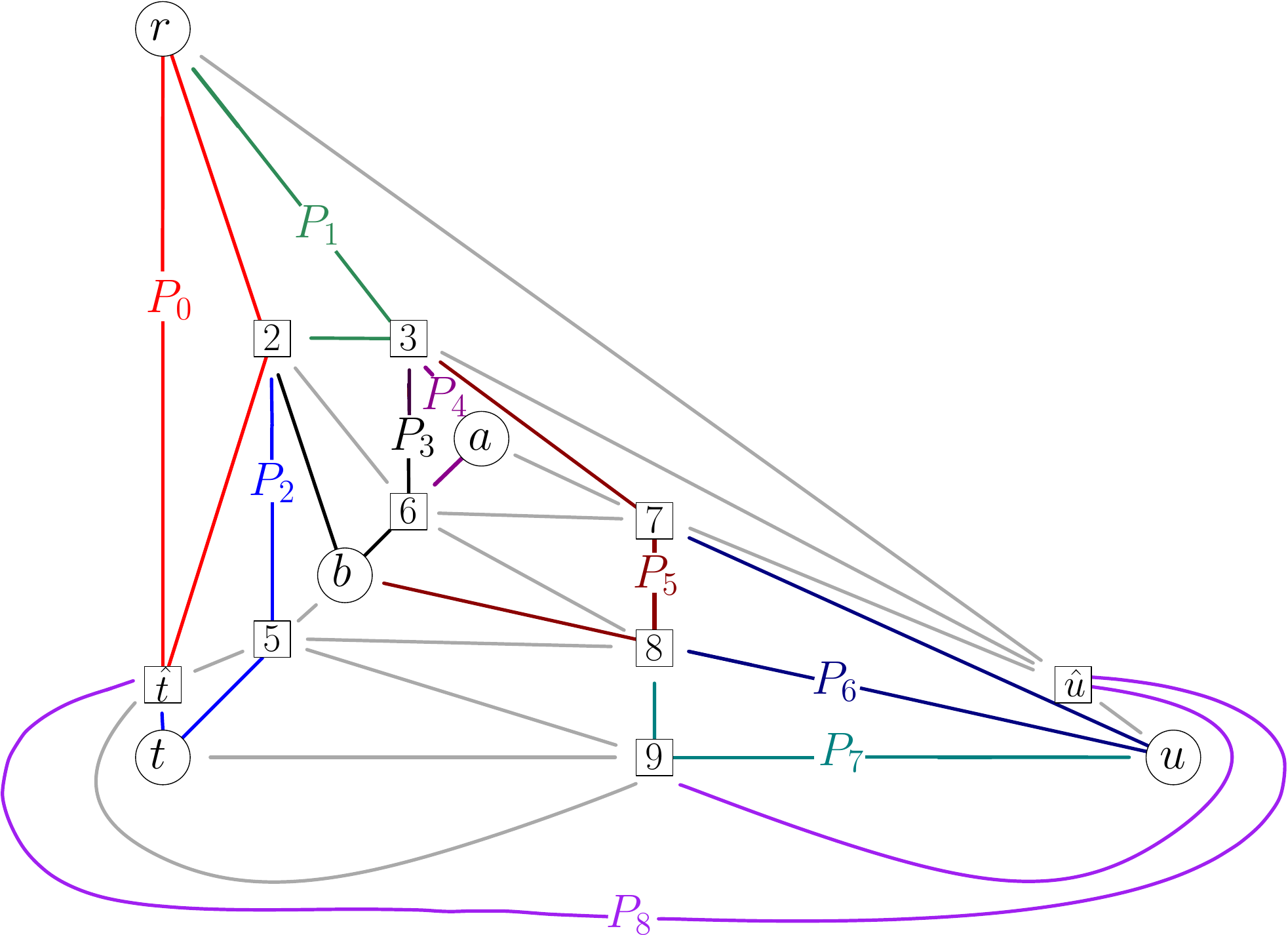}
		\label{fig:exGI_2}
	}
	\hspace{1cm}
	\subfloat[The three vertex-independent spanning trees $T_1,T_2,T_3$ of $\hat{G}$ implied by the (2,1)-order of $\hat{G}$.]{
		\includegraphics[scale=0.4]{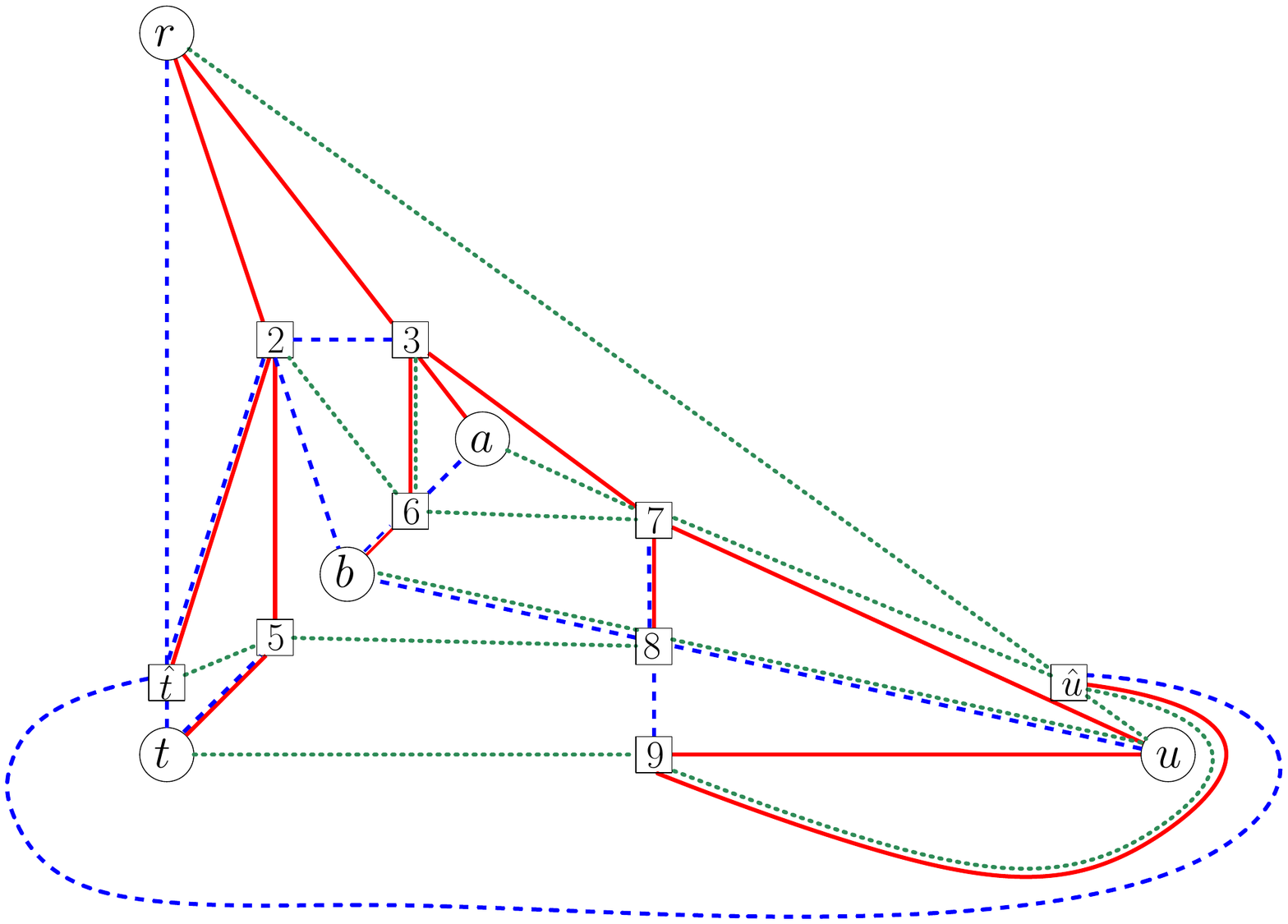}
		\label{fig:exGI_3}
	}
	\hspace{1cm}
	\subfloat[``Corresponding'' subgraphs of $T_1$, $T_2$ and $T_3$ in $G$. The red subgraph contains a cycle.]{
		\includegraphics[scale=0.6]{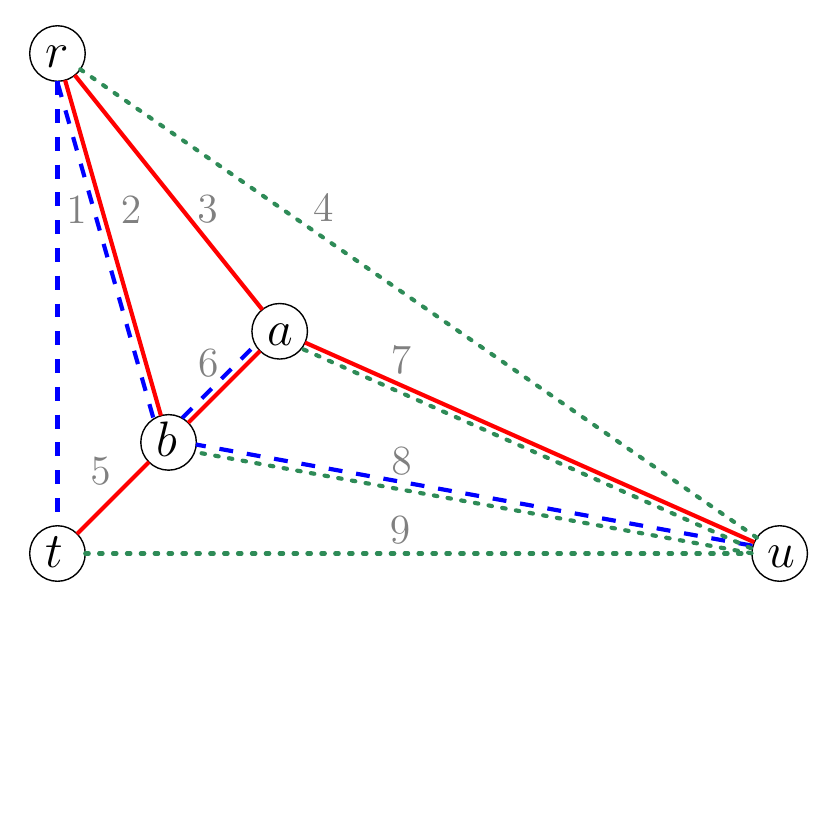}
		\label{fig:exGI_4}
	}
	\caption{The reduction~\cite{Galil1991} cannot be applied to find edge-independent spanning trees directly, as it may construct cycles.}
	\label{fig:exGI}
\end{figure}
That such a reduction could indeed be elusive, might also be argued by the fact that we still do
not know any way of reducing the existence of edge-independent spanning trees to the existence of their vertex-counterpart. In fact, a proposed such reduction turned out to be wrong.
}

\bibliographystyle{abbrv}
\bibliography{references}
\end{document}